\date{}
\documentclass[11pt]{article}
\usepackage{amsmath,amsthm,amsfonts,amssymb,amsopn,amscd}
\usepackage{color}

\DeclareMathOperator{\Tr}{Tr}

\newcommand{\ben}{\begin{equation}}
\newcommand{\een}{\end{equation}}

\def\tilde{\widetilde}

\def\a{\alpha}
\def\b{\beta}

\def\g{\gamma}
\def\Ga{\Gamma}

\def\l{\lambda}
\def\La{\Lambda}
\def\r{\rho}

\def\th{\theta}
\def\om{\omega}
\def\Om{\Omega}

\def\setminus{\smallsetminus}

\def\A{{\cal A}}
\def\B{{\cal B}}

\def\F{{\cal F}}

\def\M{{\cal M}}
\def\N{{\cal N}}
\def\R{{\cal R}}
\def\L{{\mathcal L}}

\def\H{{\cal H}}
\def\K{{\cal K}}
\def\S{{\cal S}}

\def\s{{\sigma}}

\def\l{{\lambda}}
\def\La{{\Lambda}}
\def\x{{\xi}}

\def\PSL{{{\rm PSL}(2,\mathbb R)}}

\def\S2{S^{1(2)}}

\def\Reali{\mathbb R}
\newtheorem{theorem}{Theorem}[section]
\newtheorem{lemma}[theorem]{Lemma}

\newtheorem{corollary}[theorem]{Corollary}

\newtheorem{proposition}[theorem]{Proposition}

\theoremstyle{definition} 

\theoremstyle{remark} \newtheorem{remark}[theorem]{Remark}

\def\setminus{\smallsetminus}

\def\PSL{PSU(1,1)}

\def\sp{{\mathrm {sp}}}

\def\SL2{{{\rm SL}(2,\R)}}

\def\PSL2{{{\rm PSL}(2,\Reali)}}
\def\mob{{\rm Mob}}
\def\U1{{{\rm V}(1)}}
\def\SU2{{{\rm SV}(2)}}

\def\SU{{{\rm SU}}}

\def\A{{\mathcal A}}
\def\B{{\mathcal B}}

\def\F{{\mathcal F}}
\def\H{{\mathcal H}}
\def\S{{\mathcal S}}

\def\K{{\mathcal K}}

\def\M{{\mathcal M}}
\def\N{{\mathcal N}}

\title{\Huge{Von Neumann Entropy in QFT}}
\author{{\sc Roberto Longo}\thanks{Supported by the ERC Advanced Grant 669240 QUEST ``Quantum Algebraic Structures and Models'', MIUR FARE R16X5RB55W  QUEST-NET and GNAMPA-INdAM.} \\ Dipartimento di Matematica, Universit\`a di Roma Tor Vergata \\
Via della Ricerca Scientifica, 1 - 00133 Roma, Italy
\\ Email:
{\tt longo@mat.uniroma2.it} \\ [4mm] {\sc Feng Xu}\footnote{Supported in part by NSF grant DMS-1764157.} \\
Department of Mathematics, University of California at Riverside\\
Riverside, CA 92521\\
E-mail: {\tt xufeng@math.ucr.edu}
}

\date{}

\begin{document}

\maketitle

\begin{abstract}
In the framework of Quantum Field Theory, we provide a rigorous, operator algebraic notion of entanglement entropy associated with a pair of open double cones $O \subset \tilde O$ of the spacetime, where the closure of $O$ is contained in $\tilde O$.
Given a QFT net $\A$ of local von Neumann algebras $\A(O)$,
we consider the von Neumann entropy $S_\A(O, \tilde O)$ of the restriction of the vacuum state to the canonical intermediate type $I$ factor for the inclusion of von Neumann algebras $\A(O)\subset \A(\tilde O)$ (split property).
We show that this canonical entanglement entropy $S_\A(O, \tilde O)$ is finite for the chiral conformal net on the circle generated by finitely many free Fermions (here double cones are intervals).
To this end, we first study the notion of von Neumann entropy of a closed real linear subspace of a complex Hilbert space, that we then estimate for the local free fermion subspaces.
We further consider the lower entanglement entropy $\underline S_\A(O, \tilde O)$, the infimum of the vacuum von Neumann entropy of $\F$, where $\F$ here runs over all the intermediate, discrete type $I$ von Neumann algebras. We prove that $\underline S_\A(O, \tilde O)$ is finite for the local chiral conformal net generated by finitely many commuting $U(1)$-currents.
\end{abstract}
\newpage

\section{Introduction}
Von Neumann entropy is the basic concept in quantum information and extends the classical Shannon's information entropy notion to the non commutative setting.

As is well known, a state $\omega$ on a matrix algebra $M$ is given by a density matrix $\r$, namely $\omega(T) = \Tr(\r T)$, $T\in M$. The von Neumann entropy of $\omega$ is given by
\[
S(\omega) = - \Tr(\r\log\r) \ .
\]
Entanglement entropy is a measure of quantum information by the degree of quantum entanglement of a quantum state.

Let's consider a bipartite, finite-dimensional quantum system $M=A\otimes B$, where $A$ and $B$ are matrix algebras. Given a pure state $\omega$ on the matrix algebra $M$, let $\r_A$ and $\r_B$ the density matrices associated with the restrictions of $\omega$ to $A$ and $B$ respectively on $A$ and $B$. The entanglement entropy of $\omega$ is defined as
\ben\label{EE}
- \Tr(\r_A\log\r_A) = - \Tr(\r_B\log\r_B) \ .
\een
The above definition directly extends to the case $M = \bigoplus_k A_k\otimes B_k$ is a direct sum, where $A_k$, $B_k$ are matrix algebras and the restrictions $\omega|_{A_k\otimes B_k}$ are pure (not normalised). It also extends to the infinite-dimensional case where $A_k$, $B_k$  are type $I$ factors; in this last case, however, the entanglement entropy may be infinite.

Entanglement is certainly one of the main feature of quantum physics and there is a famous, long standing debate on its interpretation (EPR paradox, Bell's inequalities, etc.). An overview of the matter lies beyond the purpose of this introduction.

The role of entanglement in Quantum Field Theory is more recent and increasingly important; it represents a piece of the quantum information framework in this subject. It appears in relation with several primary research topics in theoretical physics as area theorems, $c$-theorems, quantum null energy inequality, etc. (see for instance \cite{CCD, CH, Wit} and refs. therein).

Despite the rich physical literature on the subject, the rigorous definition of entanglement entropy in QFT is however not obvious. The point is that the von Neumann algebra $\A(O)$ associated with a double cone spacetime region $O$ is typically a factor of type $III$, so no trace exists on $\A(O)$ and one cannot naively extends the definition \eqref{EE} as one would  do with $A = \A(O)$, $B = \A(O')$, where $O$ is a double cone and $O'$ is its causal complement and $\om$ the vacuum state. Due to ultraviolet divergence, such a measure of the vacuum entanglement would always result to be infinite.
By Haag duality, that holds in much generality, $\A(O')$ is the commutant $\A(O)'$ of $\A(O)$ on the vacuum Hilbert space $\H$, so the von Neumann $\A(O)\vee\A(O')$ generated by $\A(O)$ and $\A(O')$ is equal to $B(\H)$ and cannot be naturally isomorphic  to the von Neumann tensor product $\A(O)\otimes\A(O')$.

To get rid of short distance divergences, on may however consider a slightly larger double cone $O\subset \tilde O$, namely the closure of $O$ is contained in the interior of $\tilde O$. The split property states that there is a natural isomorphism of von Neumann algebras
\[
\A(O)\vee\A(\tilde O') \simeq \A(O)\otimes\A(\tilde O')\, ,
\]
that identifies $\A(O)$ with $\A(O)\otimes 1$ and $\A(O')$ with $1\otimes \A(O')$.

The split property expresses the statistical independence of  $\A(O)$  and $\A(\tilde O')$; it was verified for the free, neutral Boson QFT case in \cite{B}.  It was studied in \cite{DL} and led to important structural features both in Mathematics and in Physics. It follows under natural, general physical requirements \cite{BW}.  It holds automatically in chiral conformal QFT \cite{MTW}. (See \cite{HO} for a discussion of its validity in topologically non trivial spacetimes).

Approaches to the entanglement entropy by means of the split property are studied in \cite{Nar, Cas, OT, Wit, HS, Fal}. In particular, \cite{HS} contains a rigorous definition via separable states and \cite{OT} considers essentially our definition below.

The split property is a local property, in fact it is equivalent to the existence of an intermediate type $I$ factor $\F$ between $\A(O)$ and $\A(\tilde O)$
\ben\label{splitA}
\A(O) \subset \F \subset \A(\tilde O) \ .
\een
A type $I$ factor $\F$ is a von Neumann algebra isomorphic to $B(\K)$, the algebra of all bounded linear operators on some Hilbert space $\K$.

We may then define the entanglement entropy of the net $\A$ associated with the double cones $O \subset \tilde O$ as the vacuum von Neumann entropy associated with the $\F$ as in \eqref{EE}; here the global systems is $B(\H)$, the factorisation is given by $\F$, namely $A = \F,$ $B = \F'$ with a tensor product decomposition
\[
\H = \H_A\otimes \H_B\, , \quad A \simeq B(\H_A)\otimes 1, \quad B \simeq 1 \otimes B(\H_B)\, ,
\]
and the pure state is the vacuum state.

This definition however depends on the choice of $\F$.
Actually, if the split property holds, there are infinitely many intermediate type $I$ factors $\F$ in \eqref{splitA}.  Yet, as shown in \cite{DL}, there is a canonical intermediate type $I$ factor $\F$, associated with the $O, \tilde O$ and the vacuum vector $\Omega$, given by the formula \ben\label{can}
\F = \A(O)\vee J\A(O)J = \B(\tilde O)\cap J\B(\tilde O)J
\een
(if the local von Neumann algebras are factors),
with $J$ is the modular conjugation of the relative commutant von Neumann algebra $\A(O)'\cap \A(\tilde O)$ associated with $\Omega$.

We then define the (canonical) entanglement entropy of $\A$ with respect to  $O,\tilde O$ as
\ben\label{EESc}
S_\A(O,\tilde O) = -\Tr(\r_\F \log\r_\F) \, ,
\een
where $\F$ is the canonical intermediate type $I$ factor \eqref{can}. Here $\Tr$ is the trace of $\F$
(namely $\F = B(\H_A)\otimes 1_{\H_B}$ and $\Tr$ corresponds to the usual trace on $B(\H_A)$) and $\r_\F$ is the vacuum density matrix relative to $\F$.

The above definition concerns a local net $\A$. If $\A$ if a Fermi net, graded locality rather than locality holds. In this case, the split property is still defined by \eqref{splitA} and the entanglement entropy by \eqref{EESc}. However, the canonical intermediate type $I$ factor is to be defined by a twisted version of formula \eqref{can}, cf. \eqref{Fcan}.

A main result in this paper is that above defined canonical entanglement entropy is finite for the chiral conformal net $\M$ generated by a complex free fermion on $S^1$. Here, double cones are intervals $I \subset \tilde I$ of $S^1$.

$\M$ is a second quantisation net and we first provide an abstract analysis on the one particle Hilbert space $\H_0$. Let $F\subset \H_0$ be a closed, real linear subspace of $\H_0$. We define the entropy of $H$ as the von Neumann entropy
\[
S(F) = S(\s)\, ,
\]
with $\s = P_F P_{F'} P_F$, where $P_F$ is the real orthogonal projection onto $F$ and $F'$ is the symplectic complement of $F'$. Note that $P_{F'} P_F$ is an ``angle operator'' (cf. \cite{D, H, RV}).
$\s$ is a bounded, real linear operator. If $S(\s) < \infty$, then $\s$ has to be a (real) trace class operator;  with $\{\l_n\}$ the (positive) proper values of $\s$, then $S(\s) = -\sum_n \l_n\log\l_n$.

It turns out that
\[
S(\s) <\infty \Longleftrightarrow S(\r_\F) < \infty \, ,
\]
where $S(\r_\F)$ is the vacuum von Neumann entropy relative to the type $I$ factor $\F$ associated with $F$ in second quantisation.
We shall choose a canonical type $I$ subspace $F$, intermediate between the subspaces real $H(I) \subset H(\tilde I)$ of $\H_0$ associated with the interval $I\subset \tilde I$, so that $S(\s)$ will be the canonical entanglement entropy \eqref{EESc}.

In the second part of this paper, we provide model analysis in order to derive the finiteness of $S(\s)$, in the one complex free fermion case.
This is the first case where the entanglement entropy \eqref{EESc} is proved to be finite, a problem that is explicit or implicit in various papers, \cite{Nar, OT}.
The same finiteness result then follows for other related nets, as in the $r$-fermion case or in the case of a real fermion.

Now, there is another natural definition of entanglement entropy:
\[
\underline S_\A(I, \tilde I) = \inf_\F S(\F) \, ,
\]
the infimum of the  the von Neumann entropy $S(\F)=-\Tr(\r_\F \log\r_\F)$ over all the intermediate type $I$ factors $\F$ or, more generally, over all the intermediate type $I$ discrete von Neumann algebras $\F$. Of course, the lower entanglement entropy $\underline S_\A(I, \tilde I)$ is bounded by the canonical entanglement entropy $S_\A(I, \tilde I)$.

We shall infer form our results that
\[
\underline S_\A(I, \tilde I) < \infty
\]
also in the case of the local net generated by $r$ free bosons on $S^1$ ($r$ copies of the current algebra net). We expect the canonical entanglement entropy to be finite in this case too, but this remains unproven in this paper.

\section{Entropy of standard subspaces}
In this first part, we discuss basic, abstract aspects concerning entropy in first and second quantisation. We go slightly beyond what is strictly needed in the second part, with the purpose of clarifying the general picture, that gives motivation for future work.

\subsection{Trace and determinants in second quantisation}\label{Trace}
Let $\H$ be an Hilbert space and
$\Ga(\H)$ (resp. $\Lambda(\H)$), the Bose (resp. Fermi) Fock Hilbert space over $\H$
\ben\label{Fock}
\Ga(\H) = \bigoplus_{n = 0}^{\infty}\Ga^n(\H)\, ,
\qquad
\La(\H) = \bigoplus_{n = 0}^{\infty}\La^n(\H)\, .
\een
If
$A\in B(\H)$ and $||A||\leq 1$ the second quantisation of
$\Ga(A)$ (resp. $\Lambda(A)$) is the linear contraction on $\Ga(\H)$ (resp. $\Lambda(\H)$) defined by
\[ 1\oplus A\oplus
(A\otimes A)\oplus(A\otimes A\otimes A)\oplus\cdots
\]
where the $A\otimes\cdots\otimes A$ acts on the symmetric part $\Ga^k(\H)$ (resp. anti-symmetric part $\Lambda^k(\H)$) of $\H\otimes\cdots\otimes \H$.

We recall the following lemma, see \cite{KL}.
\begin{lemma}\label{9}
If $A$ is selfadjoint, $0\leq A< 1$, then
\begin{gather}\label{lT}
\Tr\big(\Ga(A)\big)=\det(1 - A)^{- 1}, \qquad \Tr\big(\Lambda(A)\big) =\det(1+ A)\, ,
\\
\log\Tr\big( \Ga(A)\big)=  -\Tr\log(1 - A)\,  , \qquad \log\Tr\big( \Lambda(A)\big)= \Tr\log(1 + A)\, .\label{lT2}
\end{gather}
\end{lemma}
\begin{proof}
We may assume that $A$ has discrete spectrum, otherwise all quantities in \eqref{lT}, \eqref{lT2} are infinite.
Suppose first that $\H$ is one-dimensional, thus $A=\l$ is a
scalar $0\leq\l<1$. In the Bose case, $\Ga^n(\H)$ is then
one-dimensional for all $n$, thus we have $\Ga(A)=\bigoplus_{n=0}^{\infty}\l^n$, so
$\Tr\big( \Ga(A)\big)=\sum_{n=0}^{\infty}\l^n=(1-\l)^{-1}$.

For a general $A$,
we may decompose $\H=\bigoplus_n\H_n$ so that dim$\H_n =1$
and $A=\bigoplus_n \l_n$. Then $\Ga(\H)=\bigotimes_n^{\{\Om_n\}} \Ga(\H_n)$, where
$\Om_n$ is the vacuum vector of $\Ga(\H_n)$, and
$\Ga(A)=\bigotimes_n \Ga(A_n)$. It follows that
\[
\Tr \big(\Ga(A)\big) =\prod_n\Tr\big(\Ga(A_n)\big)=\prod_n (1-\l_n)^{-1} =\det(1-A)^{-1}.
\]
In the Fermi case, if $\H$ is one-dimensional then
$\Lambda^n(\H) =\{0\}$ if $n\geq 2$ and is one-dimensional if
$n=0,1$; if $A=\l$ we then have $\Lambda(A) = 1\oplus\l$ so  $\Tr \big(\Lambda(A)\big) = 1 +
\l$. Since, also in the Fermi case, there is a
canonical equivalence between $\Lambda\big(A\oplus B\big)$ and
$\Lambda(A)\otimes\Lambda(B)$, we have
\[
\Tr\big(\Lambda(A)\big)=\prod_n \Tr\big(\Lambda((\l_n)\big)=\prod_n(1+\l_n) =\det(1+A)\ ,
\]
where $A=\bigoplus_n\l_n$.

Concerning formulas \eqref{lT2}, notice that
\[
\det A=e^{\Tr\log A},
\]
hence by \eqref{lT2} we have
\[
\log\Tr\!\big(\Ga(A)\big) = -\log\det(1 - A)=- \Tr\!\big(\!\log(1- A)\big)
\]
and
\[
\log\Tr\!\big(\Lambda(A)\big) =\log\det(1 + A)= \Tr\!\big(\!\log(1 + A)\big).
\]
\end{proof}
As a consequence
\ben\label{Sf}
\Tr(A) < \infty \Leftrightarrow \Tr\big(\Ga(A)\big) < \infty \Leftrightarrow \Tr\big(\Lambda(A)\big)\ .
\een

\subsection{Angle between subspaces and type I property}\label{Angle}

Let $\H$ be a complex Hilbert space and $T$ a (complex) linear operator on $\H$. We shall say that $T\in \L^p$, $p\in (0,\infty)$, if there exist orthonormal sequences of vectors $\{e_k\}, \{f_k\}$ of $\H$ and a sequence of complex numbers $\l_k$ with $\sum_k | \l_k |^p <\infty$, i.e. $\{\l_k\}\in\ell^p$, such that
\begin{equation}\label{lp}
T\xi = \sum_k \l_k (e_k,\xi)f_k\ , \quad \xi\in\H\ ,
\end{equation}
thus $T\in \L^p$ iff $\Tr(|T|^p)<\infty$ and we have
$||T - T_n||\to 0$ with $T_n \equiv \sum_{k=1}^n (e_k,\,\cdot)f_k$.
We set as usual
$||T||_p\equiv \Tr(|T|^p)^\frac1p = (\sum_{k=0}^\infty |\l_k|^p)^\frac1p$. $\L^1$-operators are also called trace class operators and $\L^2$-operators Hilbert-Schmidt operators.

One can define $\L^p$ real linear operators between real Hilbert spaces analogously as in \eqref{lp}.

We may view the complex Hilbert space $\H$ as a real Hilbert with
scalar product $(\xi,\eta)_{\mathbb R} \equiv \Re(\xi,\eta)$, and
denote it by $\H_\Reali$. A complex linear operator $T:\H\to\H$ is
also a real linear operator $T:\H_\Reali \to \H_\Reali$, that we may
also denote by $T_\Reali$ if we want to emphasise that $T$ is
regarded as a real linear operator.
\begin{lemma}
$T:\H\to\H$ is $\L^p$ as a complex linear operator iff $T_\Reali:\H_\Reali\to\H_\Reali$ is $\L^p$ as a real linear operator.
\end{lemma}
\begin{proof}
Let $T:\H\to\H$ be $\L^p$ as a complex linear operator. For some orthonormal families $\{e_k\}$ and $\{f_k\}$ in $\H$ and a $\ell^p$-sequence $\{\l_k\}$ (we may assume that $\l_k\geq 0$) we have
\begin{multline*}
T\xi = \sum_k \l_k (e_k,\xi)f_k = \sum_k\l_k \Re(e_k,\xi)f_k + i\sum_k\l_k \Im(e_k,\xi)f_k \\
= \sum_k \l_k\Re(e_k,\xi)f_k + \sum_k \l_k\Re(ie_h,\xi)if_k \ ,
\end{multline*}
thus $||T_\Reali||_p^p = 2 ||T||^p_p$.

Conversely, assume that $T_\Reali\in \L^p$. For some orthogonal families $\{e_k\}$ and $\{f_k\}$ in $\H_\Reali$ we have
\[
T\xi = \sum_k \l_k\Re(e_k,\xi)f_k \ ,\quad Ti\xi = \sum_k \l_k\Re(e_k,i\xi)f_k = -\sum_k \l_k \Im(e_k,\xi)f_k \ ,
\]
thus
\[
T\xi = (1-i)^{-1}(T\xi - iT\xi) = \frac{1+i}{2}\sum_k\l_k(e_k,\xi)f_k
\]
and we have $||T||_p \leq \sqrt2\, ||T_\Reali||_p$.
\end{proof}
Let $\H$ be a complex Hilbert space and $H\subset\H$ a standard
subspace, i.e. $H$ is a closed, real linear subspace of $\H$ with
$H\cap iH =\{0\}$ and $H + iH$ dense in $\H$. $S_H, \Delta_H, J_H$
denote the usual operators associated with $H$ by the modular
theory, see \cite{L08}. We denote by $E_H(Z)$ the spectral projection of $\Delta_F$ associated with the Borel subset $Z\subset \mathbb R$, and by $\H_F(Z)$ the corresponding spectral subspace.

We shall say that $H$ is factorial if $H\cap H' = \{0\}$,
with $H'$ the real orthogonal of $iH$.
Note that $H$ is factorial iff $1$ is not an eigenvalue of
$\Delta_H$.
This is equivalent to the
Bose (resp. Fermi) second quantisation von Neumann algebra $\R(H)$  (resp. $\M(H)$)  to be a factor.

 We shall say that a factorial standard subspace $H\subset\H$ is of type if $\R(H)$ is a type $I$
factor; as we shall see in Corollary \ref{tI}, $H$ is of type $I$ iff $\M(H)$ is a type $I$ factor.

The following lemma is proved in \cite{FG} in the Bose case.
\begin{proposition}\label{PropI}
$H$ is of type $I$ iff $\Delta_H E_H(0,1)$ is $\L^1$, with $E_H(0,1)$ the spectral projection of $\Delta_H$  associated with the interval $(0,1)$.
\end{proposition}
\begin{proof}
Immediate by Corollary \ref{RF}.
\end{proof}
We denote by $P_H$ the real linear orthogonal projection from $\H_\Reali$ onto $H$. Note that
\begin{equation}
iP_{iH} = P_H i \ , \quad P_{H'} = 1 - P_{iH} \ .
\end{equation}
The following relation holds \cite{FG}:
\begin{equation}\label{PH}
P_H = \frac{1}{\Delta_H + 1} + J_H\frac{\Delta_H^{1/2}}{\Delta_H + 1}\ .
\end{equation}
To get the above formula, one easily checks that $P_H \xi = \xi$ if $S_H\xi = \xi$ and that $S_H P_H \xi = P_H \xi$.

Now, taking into account that $\Delta_{H'} = \Delta_H^{-1} = J_H \Delta_H J_H$, we have
\begin{equation}\label{PH'}
P_{H'} = \frac{\Delta_H}{\Delta_H + 1} + J_H\frac{\Delta_H^{1/2}}{\Delta_H + 1}\ .
\end{equation}
Thus the angle operator between $H$ and $H'$ is given by
\begin{align}
P_H P_{H'} &= 2\frac{\Delta_H}{(\Delta_H + 1)^2}
+ 2\frac{\Delta_H^{1/2}}{(\Delta_H + 1)^2}J_H\\
&=  2\frac{\Delta_H}{(\Delta_H + 1)^2}
+ 2\frac{\Delta_H}{(\Delta_H + 1)^2}S_H \ .
\label{AH}
\end{align}
\begin{lemma}\label{FG}\cite{FG}
Set $A_H\equiv 4 \frac{\Delta_H}{(\Delta_H + 1)^2}$.
We have
\ben\label{PPH}
P_H P_{H'}\big|_H = A_H \big|_H \ .
\een
As a consequence, $P_H P_{H'}|_H$ is $\L^p$ iff $A_H$ is $\L^p$.
\end{lemma}
\begin{proof}
Formula \eqref{PPH} follows by \eqref{AH} and the above discussion.

If $A_H$ is $\L^p$ as a complex linear operator, then $A_H$ is $\L^p$ as a real linear operator, hence
$P_H P_{H'}|_H = A_H |_H$ is $\L^p$.

Conversely, if $P_H P_{H}|_H$ is $\L^p$, then
$A_H |_H: H\to H$ is $\L^p$ as real linear operator; choose an orthonormal basis
$\{e_k\}$ and $\{f_k\}$ for $H$ w.r.t. the real part of the scalar product and $\{\l_k\}\in \ell^p$ such that $A_H\xi = \sum_k \l_k\Re(e_k,\xi)f_k$, $\xi \in H$. Since $A_H$ is complex linear
\[
A_H(\xi + i\eta) = A_H\xi + i A_H \eta
= \sum_k \l_k\Re(e_k,\xi)f_k + i\sum_k \l_k\Re(e_k,\eta)f_k \ ;
\]
as $H + iH$ is dense in the Hilbert space,  $A_H$ is thus $\L^p$ as real linear operator, hence as complex linear operator.
\end{proof}
In the following, $H$ is a factorial standard subspace of $\H$.
\begin{proposition}\cite{A}
$H$ is of type I iff $A_H$ is trace class.
\end{proposition}
\begin{proof}
Since $J_H \Delta_H J_H = \Delta_H^{-1}$, we have that $\Delta_H E_H(0,1)$ is $\L^1$ iff
$\Delta^{-1}_H E_H(1,\infty)$ is $\L^1$, thus  iff $\frac{\Delta_H}{(\Delta_H + 1)^2}$ is $\L^1$.
The proposition then follows by Prop. \ref{PropI}.
\end{proof}
\begin{corollary}\label{pp'} The following are equivalent:
\begin{itemize}
\item $H$ is of type I,
\item $P_HP_{H'}$ is $\L^2$,
\item $P_H P_{H'}|_H$ is $\L^1$,
\item $P_H P_{H'}P_H$ is $\L^1$,
\item $[P_H , i]$ is $\L^2$, where $[P_H , i] \equiv P_H i - iP_H$ ,
\end{itemize}
The eigenvalues of $P_HP_{H'}P_H$ and $\Delta_H |_{\H_H(0,1)}$ have equivalent asymptotic, in particular $\Delta_H |_{\H_H(0,1)}$ is $\L^p$ iff $P_HP_{H'}$ is $\L^{2p}$, $p>0$.
\end{corollary}
\begin{proof}
$H$ is of type $I$ iff $A_H$ is $\L^1$ iff $P_H P_{H'}|_H$ is $\L^1$ by Lemma \ref{FG}. On the other hand, $P_H P_{H'}|_H$ is $\L^1$ iff $P_H P_{H'}P_H$ is $\L^1$, thus iff $P_H P_{H'}$ is $\L^2$ as
$P_H P_{H'}P_H = ( P_{H'}P_H)^* P_{H'}P_H$. So the equivalence among the first four properties follows.

Concerning the last property, it suffices to note that, by \eqref{PH}, we have
\[
[P_H , i] = 2J_H\frac{\Delta_H^{1/2}}{\Delta_H + 1} = J_H A_H^{1/2}\ .
\]
The last assertion follows by Lemma \ref{FG}.
\end{proof}
\subsection{Type I density matrix}

Let $\H$ be a complex Hilbert space, $\Ga(\H)$ the Bose Fock space
over $\H$ and $F\subset \H$ be a standard subspace. If the von Neumann algebra $\R(F)$ on
$\Ga(\H)$ associated with $F$ (i.e. generated by the Weyl unitaries $W(h)$ with $h\in F$, see \cite{LRT}) is of type $I$, the
restriction of the vacuum state $\om$ of $B\big(\Ga(\H)\big)$ to the type $I$ factor
$\R(F)$ gives a density matrix $\r_F \in \R(F)$ with respect to the trace $\Tr_F$ of $\R(F)$
\[
\omega |_{\R(F)} = \Tr_F(\r_F\, \cdot)\, ,
\]
that we aim to analyse.

Similarly, in the Fermi case, $\r'_F$ denotes the density matrix associated with the restriction of
$\om$ to $\M(F)$, where $\M(F)$ is the von Neumann algebra on $\Lambda(\H)$ associated with $F$ (see \cite{Fo} or Sect. \ref{FFn}), assuming it to be a type $I$ factor.

We shall consider also the norm one normalisations
\[
\hat\r_F\equiv \frac{\r_F}{||\r_F||}\ ,\quad \hat\r'_F\equiv \frac{\r'_F}{||\r'_F||}\ ;
\]
thus $\r_F\equiv \frac{\hat\r_F}{\Tr_F(\hat\r_F)}$ and
$\r'_F\equiv \frac{\hat\r'_F}{\Tr_F(\hat\r'_F)}$.
\begin{lemma}\label{2dim}
Let $\H$ be a 2-dimensional complex Hilbert space and $F\equiv F_\l\subset\H$ a standard subspace with $\sp(\Delta_F) =\{\l ,\l^{-1}\}$, $\l\in(0,1)$.  We have:
\smallskip

\noindent
{\rm Bose case:}
\begin{itemize}
\item[$a)$] $\R(F)$ is a factor of type $I_\infty$.
\item[$b)$] $\sp(\hat\r_F) =\{\l^n, n= 0,1,2\dots\}^-$ and each $\l^n$, $n>1$, is an eigenvalue with multiplicity 1, while 0 is not in the point spectrum.
\item[$c)$] $\Tr_F(\hat\r_F) = \frac{1}{1-\l}$, with $\Tr_F$ the trace of $\R(F)$.
\end{itemize}
{\rm Fermi case:}
\begin{itemize}
\item[$a)$]  $\M(F)$ is a factor of type $I_2$ ($2\times 2$ matrix algebra).
\item[$b)$] $\sp(\hat\r'_F) =\{1, \l\}$.
\item[$c)$] $\Tr_F(\hat\r'_F) = 1 +\l$, with $\Tr_F$ the trace of $\M(F)$.
\end{itemize}
\end{lemma}
\begin{proof}
(Bose case).
$\H$ is the direct sum of two one-dimensional complex Hilbert spaces, $\H= \H_1\oplus \H_{-1}$, with $\Delta \xi = \l^{\pm 1} \xi$, $\xi\in \H_{\pm 1}$ and $J(\x_1 \oplus \xi_{-1}) = \bar\x_{-1} \oplus \bar\xi_1$
(with $\Delta = \Delta_F$, $J = J_F$). Then $\x_1 \oplus \xi_{-1}\in H$ iff $S\xi_1 \oplus \xi_{-1} = \x_1 \oplus \xi_{-1}$, thus iff
\[
\x_1 \oplus \xi_{-1} = \l^{-1/2} \bar\xi_{-1} \oplus \l^{1/2} \bar\xi_1 \ ,
\]
so $F = \{\x \oplus \l J\x ,\ \x\in\H_1\}$.

The real linear map $T: \H_1 \to F$, $T\x\mapsto (1 +\l^2)^{-1/2} (\x \oplus \l J\x)$ satisfies
\[
T\l^{it}\x = \Delta^{it}T\x
\]
and, by the uniqueness of the canonical commutation relations, we have an isomorphism $\Phi: \R(\H_1) \to \R(F)$ with $\Phi\big(W(\x)\big) = W(T\x)$ because
\begin{multline*}
\Im (T\xi,T\eta) = (1 +\l^2)^{-1}\Im (\x \oplus \l J\x, \eta \oplus \l J\eta) = (1 +\l^2)^{-1}( \Im (\x,\eta) + \Im (\l J\x, \l J\eta)) \\
= (1 +\l^2)^{-1}(\Im (\x,\eta) + \l^2 \Im (\l \eta, \l \x)) = \Im (\x,\eta) \ .
\end{multline*}
Now Ad$\Delta^{it}|_{\R(F)}$ is the inner automorphism implemented by the unitary $\r^{it}$ where $\r\in \R(F)$ is the positive selfadjoint element $\hat\r_F$ which is unitary equivalent to $\Gamma(\l)$. The rest is now clear.

(Fermi case.) In this case $\Lambda(\H) = \mathbb C\oplus \H \oplus \mathbb C$ is 4-dimensional, so $\M(F)$ is isomorphic to a $2\times 2$ matrix algebra. Clearly $\sp(\Delta_{\M(F)}) = \{1 , \l, \l^{-1}\}$, so $\sp(\hat\r'_F) = \{1,\l\}$ and $\Tr_F(\hat\r'_F) = 1 + \l$.
\end{proof}
Let $F\subset \H$ be a factorial standard subspace and assume that the spectrum of the modular operator $\Delta_F$ is discrete.
Let $\l_k$ be the eigenvalues of $\Delta_F |_{\H_F(0,1)}$ (with multiplicity). Then $\H$ is the direct sum of the 2-dimensional complex Hilbert spaces $\H_{\l_k}$ with corresponding direct sum decomposition $F = \bigoplus_k F_{\l_k}$ as in Lemma \ref{2dim}. Then
\[
\hat\r_F = \bigotimes_k \hat\r_{F_{\l_k}} \ ,
\]
where the infinite tensor product (w.r.t. the vacuum vectors) is convergent iff $\R(F)$ is a type $I$ factor. In this case
\ben\label{r1}
\Tr_F(\hat\r_F ) = \prod_k \Tr_{F_{\l_k}} (\hat\r_{F_{\l_k}} ) =  \prod_k (1 - \l_k)^{-1}\ .
\een
Similarly,
\ben\label{r2}
\Tr_F(\hat\r'_F ) = \prod_k \Tr_{F_{\l_k}} (\hat\r'_{F_{\l_k}} ) = \prod_k (1 + \l_k)\ .
\een
\begin{corollary}
$\sp(\Delta_{\R(F)})$ and $\sp(\Delta_{\M(F)})$ are equal to the closure of the multiplicative group generated by $\sp(\Delta_F) \setminus \{0\}$. If $\R(F)$ (resp. $\M(F)$) is of type $I$,
we have $\sp(\hat\r_F )\subset \sp(\Delta_{\R(F)})\cap [0,1]$ (resp. $\sp(\hat\r'_F )\subset \sp(\Delta_{\M(F)})\cap [0,1]$.
\end{corollary}
\begin{proof}
Immediate by the above discussion.
\end{proof}
\begin{corollary}\label{tI}
$\R(F)$ is a type $I$ factor iff $\M(F)$ is a type $I$ factor and this is the case iff $\Delta_F E_F(0,1)$ is $\L^1$.
\end{corollary}
\begin{proof}
By \eqref{r1}, \eqref{r2}, we have
\[
\Tr_F(\hat\r_F ) < \infty  \Longleftrightarrow \sum_k \l_k < \infty  \Longleftrightarrow \Tr_F(\hat\r'_F ) < \infty
\]
and the corollary follows because $\Tr\!\big(\Delta_F E_F(0,1)\big) =\sum_k \l_k$.
\end{proof}
We make explicit the following direct consequence.
\begin{corollary}\label{RF}
Let $F\subset\H$ standard subspace with $\R(F)$ a type $I$ factor. If we
identify $\R(F)$ with $B(\K)$ with $\K$ a Hilbert space, then
$\hat\r_F$ is unitarily equivalent to $\Gamma(\Delta_F |_{\H_F(0,1)})$ on $\Ga(\H_F(0,1))$. Similarly,
$\hat\r'_F$ is unitarily equivalent to $\Lambda(\Delta_F |_{\H_F(0,1)})$.
\end{corollary}
\begin{proof}
As in Lemma \ref{2dim},
$\Ga(\Delta_F |_{\H_F(0,1)}) = \bigotimes_k \Ga(\l_k)$ which is unitarily equivalent to $\hat\r_F = \bigotimes_k \hat\r_{F_k}$.
The Fermi case is analogous.
\end{proof}

\subsection{Second quantisation entropy}
With $\r$ a positive, non-singular selfadjoint linear operator, the {\it von Neumann entropy} of $\r$ is defined by
\ben\label{vNe}
S(\r) \equiv -\Tr(\r\log\r)\ .
\een
Here, $\r$ is not assumed to have trace one, note however that for $\l > 0$
\[
S(\r) < \infty \Longleftrightarrow S(\l\r) < \infty\, .
\]
The following proposition (Fermi case) is equivalent to \cite[Lemma 3.3]{CCS}. 
\begin{proposition}
Let $A$ be  operator, $0\leq A< 1$. The von Neumann entropy of the Bose and Fermi second quantisation of $A$ is given respectively by
\ben\label{SG}
S\big(\Ga(A)\big) = -\Tr\left(\frac{A}{1-A}\log A\right) \Tr\big(\Ga(A)\big)
= -\Tr\left(\frac{A}{1-A}\log A\right)\det(1- A)^{- 1} \ ,
\een
\ben\label{SG2}
S\big(\Lambda(A)\big)= -\Tr\left(\frac{A}{1+A}\log A\right) \Tr \big(\Lambda(A)\big)
= -\Tr\left(\frac{A}{1+A}\log A\right) \det (1 +A) \ .
\een
\end{proposition}
\begin{proof}
We use the formula
\[
-S(\r) = \frac{d}{d\a}\Tr(\r^\a)\Big|_{\a =1}\ ,
\]
hence we may compute $S(\r)$ by the logarithmic derivative
\[
-S(\r) = \left(\frac{d}{d\a}\log \Tr(\r^\a)\Big |_{\a =1}\right)\Tr(\r)\ .
\]
Since $\Ga(A^\a) = \Ga(A)^\a$, we have
\begin{align*}
S\big(\Ga(A)\big) &= -\left(\frac{d}{d\a}\log \Tr\big(\Ga(A^\a)\big)\big|_{\a =1}\right)\Tr\big(\Ga(A)\big)
= \left(\frac{d}{d\a}\Tr\big(\log (1-A^\a)\big)\big|_{\a =1}\right)\Tr\big(\Ga(A)\big)
\\ &= \Tr\Big(\frac{d}{d\a}\log \big(1- A^\a)\Big)\Big|_{\a =1}\Tr\big(\Ga(A)\big)
=- \Tr\Big(\frac{A^\a}{1-A^\a}\log A\Big)\Big|_{\a =1}\Tr\big(\Ga(A)\big)\\
&= -\Tr\Big(\frac{A}{1-A}\log A\Big)\Tr\big(\Ga(A)\big)\ ;
\end{align*}
the second equation in \eqref{SG} follows by \eqref{lT}.

Similarly,
\begin{align*}
S\big(\Lambda(A)\big) &= -\left(\frac{d}{d\a}\log \Tr\big(\Lambda(A^\a)\big)\big|_{\a =1}\right)\Tr\big(\Lambda(A)\big)
= -\left(\frac{d}{d\a}\Tr(\log (1+A^\a))\big|_{\a =1}\right)\Tr\big(\Lambda(A)\big)
\\ &= -\Tr\Big(\frac{d}{d\a}\log \big(1 + A^\a)\Big)\Big|_{\a =1}\Tr\big(\Lambda(A)\big)
=- \Tr\Big(\frac{A^\a}{1+A^\a}\log A\Big)\Big|_{\a =1}\Tr\big(\Lambda(A)\big)\\
&= -\Tr\Big(\frac{A}{1+A}\log A\Big)\Tr\big(\Lambda(A)\big) \ .
\end{align*}
The second equations in \eqref{SG}, \eqref{SG2} follow by \eqref{lT}.
\end{proof}
With $A$ as above, note that $\Tr(A)$ is finite iff $\det(1 + A)$ is finite and
iff $\det(1 - A)^{-1}$ is finite. So we have:
\begin{corollary}\label{S(a1)}
$S(A)$ is finite iff $S\big(\Ga(A)\big)$ is finite and iff $S\big(\Lambda(A)\big)$ is finite.
\end{corollary}
\begin{proof}
Clearly,
\[
S(A) < \infty \Rightarrow \Tr(A) < \infty \ .
\]
Then by \eqref{SG} we have
\[
S\big(\Ga(A)\big) < \infty  \Rightarrow -\Tr\left(\frac{A}{1-A}\log A\right) < \infty
\Leftrightarrow S(A) < \infty \ .
\]
The converse implication
$S(A) < \infty \Rightarrow S\big(\Ga(A)\big) < \infty$
follows again by \eqref{SG} and by \eqref{Sf}.

The equivalence $S(A) < \infty \Leftrightarrow S\big(\Lambda(A)\big) < \infty$
is analogously obtained by \eqref{SG} and \eqref{Sf}.
\end{proof}
Let $F\subset \H$ be a factorial, type $I$ standard subspace of $\H$. We define the {\it entropy $S(F)$ of $F$} as
\ben\label{SF}
S(F) = S(P_F P_{F'} P_F)
\een
and the Bose (resp. Fermi) {\it entanglement entropy} of $F$ to be the von Neumann entropy
$S(\r_F)$ (resp. $S(\r'_F)$).
As above, $\r_F$ is the density matrix of the restriction of the vacuum state to $\R(F)$ (resp. $\M(F)$). We have
\[
S(\r_F) = S(\r_{F'})\ , \qquad  S(\r'_F) = S(\r'_{F'})\ .
\]
\begin{corollary}\label{S(a)}
\[
S(\r_F) < \infty \Longleftrightarrow S(\r'_F) < \infty \Longleftrightarrow S(F) < \infty\ .
\]
\end{corollary}
\begin{proof}
By Corollary \ref{RF}, $\hat\r_F$ is unitarily equivalent to $\Ga(P_F P_{F'} P_F)$
(after an identification of $\R(F)$ with $B(\K)$); and an analogous unitary equivalence of $\hat\r'_F$ with  $\Lambda(P_F P_{F'} P_F)$ holds.
So the statement follows by Cor.   \ref{pp'} and Cor. \ref{S(a1)}.
\end{proof}

\subsection{Split inclusions}\label{sectsplit}
Recall that an inclusion of von Neumann algebras $\N_1\subset \N_2$ is said to be split if there exists an intermediate type $I$ factor $\F$, so $\N_1 \subset \F \subset \N_2$ \cite{DL}, see
this reference for more on these inclusions. In this section we begin to study split inclusions of standard subspaces.

Let $\H$ be a Hilbert space.
With $K\subset\H$ a closed, real linear subspace, we denote as above by $P_K$ the real orthogonal projection onto $K$.
Let $K\subset H$ be an inclusion of standard subspaces of $\H$. We shall say that $K\subset H$ is {\it split} if the corresponding inclusion of von Neumann algebras $\R(K)\subset \R(H)$ on the Bose Fock space is split.
\begin{proposition}\label{split}
Let $K\subset H$ be an inclusion of standard factorial subspaces of $\H$ with $K'\cap H$ standard. The following are equivalent:'
\begin{itemize}
\item[$(i)$] $K\subset H$ is split,
\item[$(ii)$] There exists an intermediate type I standard space $F$: $K\subset F\subset H$,
\item[$(iii)$] $F\equiv \overline{K + J_{K'\cap H}K} = H\cap J_{K'\cap H}H$ is a canonical type $I$ intermediate subspace.
\end{itemize}
\end{proposition}
\begin{proof}
Clearly $(iii)\Rightarrow (ii)\Rightarrow (i)$ and we have to show that $(i) \Rightarrow (iii)$.
Assuming then $(i)$, $\R(K)\subset \R(H)$ is split by definition. The vacuum vector $\Om$ is cyclic for the relative commutant $\R(K)'\cap
R(H) = \R(K')\cap \R(H) = \R(K'\cap H)$ because $K'\cap H$ is standard by assumption.
By \cite{DL}, the von Neumann algebra
\ben\label{canI}
\F = \R(K)\vee J\R(K)J = \R(H)\cap J\R(H)J
\een
is a canonical intermediate type $I$ factor. Here $J$ is the modular conjugation of $\R(K')\cap \R(H)$ w.r.t. $\Om$.

Now, $J = \Ga(J_{K'\cap H})$, the second quantisation of the modular conjugation $J_{K'\cap H}$ of $K'\cap H$. Thus
\ben\label{canF}
\F = \R(K)\vee J\R(K)J = R(K)\vee \R(J_{K'\cap H}K) = \R(F) \, ,
\een
where $F = \overline{K + J_{K'\cap H}K}$ is an intermediate type $I$ factor between $K$ and $H$, and $F$ if of type $I$ because $\R(F) = \F$ is of type $I$.
\end{proof}
Let $K\subset H$ be a split inclusion of standard subspaces of $\H$ with $K'\cap H$ standard. By Prop. \ref{split}, there is a canonical intermediate type $I$ factor $\R(F)$ between $\R(K)$ and $\R(H)$, with $F$ given by \eqref{canF}. This is indeed the canonical intermediate type $I$ factor for $\R(K)\subset\R(H)$ associated with the vacuum vector.

By Corollary \ref{tI}, the von Neumann algebra $\M(F)$ in the Fermi quantisation of $F$ is also a canonical intermediate type $I$ factor between $\M(K)$ and $\M(H)$, that we shall study in Section \ref{PartII}, where the standard subspaces are associated with intervals of $S^1$ (free Fermi net in first quantisation). This is a of canonical intermediate type $I$ factor for $\M(K)\subset\M(H)$ associated with the vacuum vector, here constructed by a graded version \eqref{Fcan} of formula \eqref{canI}, where the relative commutant is replaced by the graded relative commutant.

Before ending this section, we note the following proposition.
\begin{proposition}
Let $K\subset H$ be an inclusion of standard factorial subspaces of $\H$ with $K'\cap H$ standard.
If $K\subset H$ is split then $P_KP_{H'}$ and $P_KJP_{H'}$ are $\L^2$, with $J = J_{K'\cap H}$.
\end{proposition}
\begin{proof}
Let $F$ be an intermediate type $I$ subspace.
By Cor. \ref{pp'} $P_F P_{F'}P_F$ is trace class. As $P_{H'} < P_{F'}$, we have $P_F P_{H'}P_F < P_F P_{F'}P_F$, so $P_F P_{H'}P_F$ is trace class. By multiplying the latter by $P_K$ from the left and from the right we see that $P_K P_{H'} P_K$ is trace class. Equivalently,  $P_{H'}P_K$ is $\L^2$.

With $F$ is the canonical intermediate type $I$ subspace, we have $JK\subset F$, so $JK\subset H$ is split. By the above argument, $P_{JK}P_{H'}$ is thus $\L^2$. As $P_{JK} = JP_K J$, we then have that $JP_K J P_{H'}$, hence $P_K J P_{H'}$, is $\L^2$.
\end{proof}
Let $K\subset H$ be an inclusion of standard factorial subspaces of $\H$ with $K'\cap H$ standard.
By the above proposition, and its proof, we might expect that $P_KP_{H'}$ is $\L^p$ iff $P_F P_{F'}$ is $\L^p$, where $F$ is the canonical intermediate type $I$ subspace. A relation of this kind would be useful for entropy computations.
\section{Entanglement entropy for conformal nets}
\label{PartII}
In this second part, we are dealing with QFT nets of standard subspaces and of von Neumann algebras on $S^1$ and we state our definitions in this setting, although they directly extends to the case of net on higher dimensional spacetimes.

\subsection{Free Quantum Field case}
Let $H: I\subset S^1 \mapsto H(I)$ be a $SL(2,\mathbb R)$-covariant net of standard subspaces of a Hilbert space $\H$ on the interval of $S^1$. We assume that $H$ is local or twisted-local.

Given intervals $I \subset \tilde I$ of $S^1$, with $\bar I$ contained in the interior of $\tilde I$, and a factorial, type $I$ standard subspace $F$ with $H(I) \subset F \subset H(\tilde I)$,
we set
\[
S_H(I,\tilde I ; F) := S(F) \, ,
\]
with $S(F)$ the von Neumann entropy of $F$ defined in  \eqref{SF}. As seen in Sect. \ref{sectsplit}, there is a canonical choice for $F$. We set
\[
S_H(I,\tilde I) = S_H(I,\tilde I; F) \, ,
\]
with $F$ the canonical intermediate type $I$ factorial standard subspace.

Analogously, let $\A$  be $SL(2,\mathbb R)$-covariant net of von Neumann algebras on $S^1$. We assume that $\A$ is either local or twisted-local and the split property holds.

Let $\F$ be an intermediate type $I$ factor between $\A(I)$ and $\A(\tilde I)$, and $\r_\F\in\F$ be the density matrix (with trace one) associated with the restriction of the vacuum state to $\F$. We set
\ben\label{AH1}
S_\A(I,\tilde I; \F) := S(\r_\F)\   .
\een
Now, consider the net of von Neumann algebras
$\A(I)$ associated with the Bose (resp. Fermi) second quantisation of $H$, and set
\[
S_\A(I, \tilde I) := S_\A(I,\tilde I; \F)\, ,
\]
with $\F$ the type $I$ factor associated with the Bose (resp. Fermi) second quantisation of the canonical intermediate type $I$ standard subspace $F$.

Then, by Corollary \ref{S(a)}, we have
\ben\label{AH2}
S_\A(I,\tilde I) < \infty \Longleftrightarrow S_H(I,\tilde I) < \infty \, .
\een
We may also consider the lower entropy
\[
\underline S_\A(I, \tilde I) := \inf_\F S_\A(I,\tilde I; \F)\, ,
\]
where the infimum is taken over all intermediate type $I$ factor $\F$ or, more generally, over all intermediate, discrete type $I$ von Neumann algebras (countable direct sum of type $I$ factors).

Of course
\[
\underline S_\A(I, \tilde I) \leq S_\A(I,\tilde I) \, .
\]
In some case, it may be easier to get a bound for $\underline S_\A(I, \tilde I)$ rather than for $S_\A(I, \tilde I)$, cf. Sect. \ref{bos}.

\subsection{Free Fermi nets}
\label{FFn}

We refer the reader to Section 3 of Chapter I and Section 13 of
Chapter II of \cite{Was} for more details.

Let $\H$ be a complex Hilbert space and $K\subset\H$ a standard
subspace of $\H$. If $\xi\in\H$, we denote by $a(\xi)$ the
associated exterior multiplication operator on the Fermi Fock space
over $\H$ \eqref{Fock} (creation operator) and $c(\x) = a(\x) + a^*(\x)$ the
Clifford multiplication.

If $T$ is a bonded linear operator on $\H$
with $||T||\leq 1$, we denote by $\La(T)$, as in Sect. \ref{Trace}, the bounded linear
operator on $\La(\H)$ such that $\La(T) |_{\La^n(\H)} = T\wedge T\cdots\wedge T$, the $n$-fold tensor product of
$T$ on $\La^k(\H)$. Note that $||\La(T)||\leq 1$.

The Klein transform
$k:\La(\H)\to\La(\H)$ is the linear operator such that $k = 1$ on
the even part  and $i$ on the odd part of $\La(\H)$.

Let $\M(K)$ be the von Neumann algebra generated by $\{c(\x), \x\in \H\}$ and $J_K$, $\Delta_K$ the modular conjugation and operator of $\M(K)$ with respect to the vacuum vector $\Om$. Then
\[
J_K = k^{-1} \La(ij_K)\, , \qquad \Delta_K^{it} = \La(\delta_K^{it})\, ,
\]
where $j_K$ and $\delta_K$ are the modular conjugation and operator of the standard subspace $K$, see \cite{A} or \cite{Was}
(in this section we use use small letters to denote modular operators and conjugations on the one-particle Hilbert space, and capital letter on the fermi Fock Hilbert space).

We have
\[
\M(K)' = J_K \M(K) J_K = k^{-1}\M(K^\perp)k\, ,
\]
with $K^\perp$ the orthogonal of $K$ with respect to the real part of the scalar product; $\M(K^\perp)$ is the graded commutant of $\M(K)$. Note that the formulas for the orthogonal projections:
\[
j_K K = K' = (iK)^\perp,\qquad P_{K'} = 1 - P_{iK} = 1 +  i P_K i = -i P_{K^\perp} i \ .
\]
It follows that
\[
P_{i j_K K} = -i P_{(iK)^\perp} i = -i (1 - P_{iK})i = P_{K^\perp} \ .
\]
We shall denote $\hat j_K = i j_K$. Then
\begin{equation}\label{P}
P_{\hat j_K K} = P_{K^\perp} \ .
\end{equation}
From now on we specialise to $\H = L^2(S^1 ; \mathbb C)$, with
the complex structure on $\H$ given by $i(2P -1)$, with $P$ is the
projection onto the Hardy space.

The net of standard subspaces on $S^1$ is given by
\[
I \mapsto L^2(I)\equiv L^2(I;\mathbb C)\ .
\]
If $K = L^2(I)$ with $I$ the upper
semicircle, then $\hat j_K$ is given by Theorem in Section 14 of
\cite{Was}:
\begin{equation}\label{jhat}
\hat j_K (f) = \frac1z f(z^{-1}) \ .
\end{equation}
\subsubsection{Explicit formula in a special case}

On $S^1$, we consider the following four ``symmetric intervals"
\begin{align}\label{2f}
I_1 =\big\{ e^{i\th} : 0 <\th< \frac{\pi}2\big\}, \quad  I_2 &=\big\{ e^{i\th} : - \frac{\pi}2 <\th< 0\big\}, \\
-I_1 =\big\{ e^{-i\th} : 0 <\th<\frac{\pi}2\big\}, \quad  -I_2 &=\big\{
e^{-i\th} : - \frac{\pi}2 <\th<0\big\}.
\end{align}
We note that $I_1^2 = (-I_1)^2 = I = \{ e^{i\th} : 0 \leq\th\leq \pi \}$, where
$I_1^2 = \{ z^2 : z \in I_1\}$. We denote by $\chi_1, \chi_2, \chi_{-1},\chi_{-2}$ the characteristic functions of $I_1,  I_2,  -I_{1}, - I_{2}$ respectively.

Now let $K = L^2(I_1) \oplus  L^2(- I_1)$.
Our goal is to have an explicit formula for $\hat j_K$: this is
based on the formula for the modular operator for disjoint intervals
in the free fermion case in \cite{LMR}, \cite{RT} and \cite{Cas}.

We will start with the following linear isomorphism from
$\H\oplus \H$ to $\H$ which is inspired by Prop. 3  in \cite{RT}
\ben
\begin{cases}\label{3f}
\b \big((\phi_1(z),\phi_2(z))\big) = \psi (z)=\phi_1(z^2) + z\phi_2(z^2) \\
\phi_1(z^2) = \frac12\big(\psi(z) + \psi(-z)\big) \\
\phi_2(z^2) = \frac{1}{2z}\big(\psi(z) - \psi(-z)\big)
\end{cases}
\een where $(\phi_1 , \phi_2) \in \H\oplus \H$, $\psi$ is in $\H
= L^2(S^1 ; \mathbb C)$.
In terms of basis one can check from the above definition that
\[
\b(z^n,0)=z^{2n}, \b(0,z^{n})=z^{2n+1}
\]
for all integers $n$.

Note that $\b P\oplus P = P\b$ where $P$ is the projection onto the
Hardy space.

The inverse of $\b$ is given by \ben\label{4f} \b^{-1}\big(\psi(z)\big) =
\big(\phi_1(z) , \phi_2(z)\big) \ .
\een
Now, for the free fermion net associated with $\phi_1$, if $K =
L^2(I)$ with $I$ the upper semicircle, we have $\hat j_K
(f) = \frac1z f(z^{-1})$ by \eqref{jhat}.

Using \eqref{3f} and \eqref{4f}, we can transform the diagonal
action of $\hat j_K$ on $\H\oplus \H$ to an action of  $j$ on
$\psi(z)$ such that $j = \b \hat j_K \b^{-1}$.

We arrive at the following \ben\label{8f} (jf)(z) =
\left(\frac{1}{2} + \frac1{2z^2}\right)f(z^{-1}) -
 \left(\frac{1}{2} - \frac1{2z^2}\right)f(-z^{-1}) \, .
\een
Note that $j$ maps $L^2(I_1)$ to $L^2(I_2 \cup  - I_2)$. We will denote by $M_I$ the multiplication operator by $\chi_I$.

By definition, $M_{I_1}$ and $jM_{I_1}j$ are orthogonal projections whose ranges $L^2(I_1)$ and  $jL^2(I_1)$ are also orthogonal.
Hence $P_{12} := M_{I_1} + jM_{I_1}j$ is a projection.

Note that $L^2(I_1)\oplus jL^2(I_1)$ is the canonical type $I$ standard subspace which is intermediate between $L^2(I_1)$ and
$L^2(I_1 \cup I_2 \cup - I_2)$.
\begin{theorem}\label{th1}
\[
P_{12} = M_g + M_h R \quad {\rm on}\ L^2(S^1) \, ,
\]
where
\begin{gather}
g(z) = \frac1{4z^2}(z^2 + 1)^2 \chi_2 - \frac1{4z^2}(z^2 -1)^2 \chi_{-2} + \chi_1\, ,\\
h(z) = \frac1{4z^2}( z^4-1)(\chi_2 - \chi_{-2})\, ,\\
R(f)(z) = f(-z)\, .
\end{gather}
\end{theorem}
\begin{proof}
By \eqref{8f} we have \ben jM_{I_1} j f = \Big( \frac1{4z^2}(z^2 +
1)^2 \chi_2 - \frac1{4z^2}(z^2 -1)^2 \chi_{-2} \Big) f(z) + \Big(
\frac1{4z^2}( z^4-1) \chi_2  - \chi_{-2} \Big) f(-z) \een and the
theorem follows.
\end{proof}
It is now useful to examine the Fourier modes of $g$ and $h$.
\begin{proposition}\label{prop1}
Consider the Fourier series $g(z) = \sum_n g_n z^n$ and $h(z) =
\sum_n h_n z^n$ where $g,h$ are as in Th. \ref{th1}. Then $|g_n| =
O(n^{-3})$ and $|h_n| = O(n^{-2})$.
\end{proposition}
\begin{proof}
Consider the Fourier expansion $\chi_{I_1}(z) = \sum_n \g_n z^n$. Then
\[
\g_n = \frac1{2\pi}\int_0^{\frac{\pi}2} e ^{-in\th}d\th =
\frac1{2\pi}\frac{i}{n} (e^{-in \pi/2} - 1)\, .
\]
Thus
\[
g(z) = \sum_p\Big(\big(\frac14\g_{2-p} +\frac12\g_{-p}+
\frac14\g_{-2-p}\big) - \big(\frac14\g_{2-p} -\frac12\g_{-p}+
\frac14\g_{-2-p}\big)(-1)^p + \g_p\Big)z^p\, ,
\]
so $g_n =0$ when $n$ is even as $\g_n = - \g_{-n}$ if $n\geq 0$ is even.

When $n$ is odd
\begin{align}
g_n  &= \frac12(\g_{2-n}  + \g_{-2-n}) + \g_n\\
&= \frac{i}{2\pi}(e^{\frac{-i\pi n}2} - 1)\Big( \frac1n -
\frac12\big(\frac1{n-2} + \frac1{n+2}\big)\Big) =
\frac{i}{2\pi}(e^{\frac{-i\pi n}2} - 1)\frac{-4}{n(n^2 - 4)} =
O(n^{-3})\, .
\end{align}
As $h(z) = \sum_n (-\g_{-2-n} + \g_{2-n})z^n$, we have
\[
h_n  = -\g_{-2-n} +\g_{2-n} = O(n^{-2})\, .
\]
\end{proof}

\subsubsection{General symmetric interval case}\label{general}
On $S^1$, we consider the following general four ``symmetric
intervals"
\begin{align}\label{9f}
I_1 =\{ e^{i\th} : 0 <\th<\phi\}\, , \quad  I_2 &=\{ e^{i\th} : \phi-\pi <\th< 0\}\, , \\
-I_1 =\{ e^{-i\th} : 0 <\th< \phi\}\, , \quad  -I_2 &=\{ e^{-i\th} :
\phi-\pi <\th< 0\}\, , \quad 0<\phi<\pi\, .
\end{align}
We note that all results in this section simplify in the previous
section when $\phi= \pi/2$.

Denote by $I_0:=\{ e^{i\th} : 0 <\th< 2\phi \}$. We shall
consider the action of $SU(1,1)$ on $S^1$ which is given by
$z\rightarrow \frac{az +b}{\bar{b} z+ \bar{a}}$ with $|a|^2-|b|^2=
\pm 1$. The M\"obius group $\mob$ is the subgroup of $SU(1,1)$ of elements with
determinant $|a|^2-|b|^2= 1.$ The action $z\rightarrow \frac{1}{z}$
is orientation reversing.

If $m(z)= \frac{az +b}{\bar b z+ \bar a}$, the unitary action of $m$ on
$S^1$ is given by (See Section $4$ of \cite{Was})
\ben\label {10f}
(U_m  f) (z) = (\alpha -\bar\beta z)^{-1} f(m^{-1}z) \, .
\een
Since $ (a -\bar{b} z)^{-1}$ is holomorphic for $|z|<1$ and
$|a|>|b|$, $U_m$ commutes with the Hardy space projection $P$. The
flip map $(F_1 f)(z)= \frac{1}{z}f(\frac{1}{z})$ clearly satisfies
$PF_1P= 1-P$. By combining this we get an action of $SU(1,1)$ on
$\H$ which is of the form
\ben\label {11f}
(U_m f) (z) =  \alpha_m(z) f(m^{-1}z) \, ,
\een
 where $m
(z)= \frac{az +b}{\bar b z+ \bar a}, \alpha_m(z)= (\alpha -\bar\beta
z)^{-1}$.

Let $m\in \mob$ be such that $m I_0$ is the upper half circle. Let
$m_1= m^{-1} F_1 m$. It is straightforward to see that
\ben\label{12f} m_1(e^{i\phi})= \frac{z(1+e^{2i\phi})/2 - e^{2i\phi}}{z-
(1+e^{2i\phi})/2 } \, .
\een
Define
\ben\label {13f}
(F_0  f) (z) =  \alpha_{m_1}(z) f(m_1^{-1}z)\, .
\een
Using \eqref{3f} and \eqref{4f}, we can transform the diagonal
action of $F_0$ on $\H\oplus\H$ to an action $j$ on $\H$ such
that $j  = \b F_0 \b^{-1}$.

We then arrive at the following expression
\ben\label{14f} (jf)(z) =
\alpha_{m_1}(z^2) \left(\frac{1}{2} + \frac{z}{2u}\right)f(u) +
\alpha_{m_1}(z^2)
 \left(\frac{1}{2} - \frac{z}{2u}\right)f(-u)\, ,
\een
where $u^2= m_1(z^2)$.

Note that $j$ maps $L^2(I_1)$ to $L^2(I_2 \cup  - I_2)$.  We will
denote by $M_I$ the multiplication operator by $\chi_I$, the
characteristic function of interval $I$.

By definition, $M_{I_1}$ and $jM_{I_1}j$ are orthogonal projections
whose ranges $L^2(I_1)$ and  $jL^2(I_1)$ are also orthogonal. Hence
$P_{12} := M_{I_1} + jM_{I_1}j$ is a projection.

Note that $L^2(I_1)\oplus jL^2(I_1)$ is the canonical type $I$
standard subspace which is intermediate between $L^2(I_1)$ and
$L^2(I_1 \cup I_2 \cup - I_2)$.
\begin{theorem}\label{th11}
\[
P_{12} = M_g + M_h R \quad {\rm on}\ L^2(S^1)\, ,
\]
where
\begin{gather}
g(z) = \frac{(u+z)^2}{4uz}  \chi_{I_1}(u) -  \frac{(u-z)^2}{4uz} \chi_{I_1}(-u) + \chi_{I_1}(z) \, ,\\
h(z) =\frac{z^2-u^2}{4uz}( \chi_{I_1}(u)- \chi_{I_1}(-u)) \, ,\\
(R(f)(z) = f(-z),  u^2= m_1(z^2)\, .
\end{gather}
\end{theorem}
\begin{proof}
This follows directly from  \eqref{14f}.
\end{proof}
To get a generalisation of  Prop. \ref{prop1} we need the following
lemma:

\begin{lemma}\label{lemmaprop2}
Let $f(z)$ be a continuous function, where $z=e^{i\theta}$ and the
support of $f$ lies in   $\theta \in [a,b] \subset [0,2\pi].$
Suppose that $f''$ exists except for finitely many points in $[a,b], $
and $|f''(z)|\leq M < \infty$. Suppose that  $f= \sum_n f_n z^n,
z=e^{i\theta}$. Then $f_n= O(n^{-2}).$
\end{lemma}
\begin{proof}
First suppose that  $f''$ exists except at $a,b$. We have
\[
f_n =
\frac{1}{2\pi} \int_0^{2\pi} f(z) z^{-n-1}dz =
\lim_{\epsilon\rightarrow 0^+} \frac{1}{2\pi}
\int_{a+\epsilon}^{b-\epsilon} f(z) z^{-n-1}dz \ .
\]
By doing integration by parts once and use the continuity of $f$ at
$a,b$  we get
$$f_n =
\frac{1}{n+2}\frac{1}{2\pi}\lim_{\epsilon\rightarrow 0^+}
\int_{a+\epsilon}^{b-\epsilon} f'(z) z^{-n-2}dz\, .
$$
Integration by
part one more time and use our assumption we conclude that
 $f_n= O(n^{-2}).$ In general if $f''$ does not exist for $z_1,...,
 z_k$, we can divide the interval $[a,b]$  into a union of finitely many intervals
 where $f''$ exists on each interval except at end points, and apply
 the argument above.
\end{proof}
We note that continuity assumption in the above Lemma is crucial.
For example $\chi_{I_1}$ verifies the assumption of the above
Lemma except continuity, and the conclusion of the Lemma does not
hold for  $\chi_{I_1}$.
\begin{proposition}\label{lp2}
Consider the Fourier series $g(z) = \sum_n g_n z^n$ and $h(z) =
\sum_n h_n z^n$ where $g,h$ are as in Th. \ref{th11}.  Then $|g_n| =
O(n^{-2})$ and $|h_n| = O(n^{-2})$.
\end{proposition}
\begin{proof}
It is sufficient to check that $g$ and $h$ verify Lemma
\ref{lemmaprop2}. Let us check this for $g$. First let us check
$g(z)$ is continuous. It suffices to do this at the end points of
intervals $I_1, I_2, -I_2$. At these endpoints by construction
$u=\pm z$, and it is easy to see that $g$ is continuous by using the
fact that $g$ is invariant under $u \rightarrow -u$. For an example
let us show that the right limit and left limit of $g$ are equal to
each other at $e^{i\phi}.$  Let $z=e^{i \theta}$ and suppose that
$\theta \rightarrow \phi^-$. Note that $u\in I_2$ or $u\in -I_2$,
hence $g(z)=1$. When  $\theta \rightarrow \phi^+$, since $g$ is
invariant under $u \rightarrow -u$, we can assume that $u\in I_1$
and $u\rightarrow e^{i\phi}.$ By the formula for $g(z)$ it follows
that  $\lim_{\theta \rightarrow \phi^+} g(z)=1$. Similarly one can
show continuity of $g$ at all other end points.

Let us show that $g''$ exists and is bounded except possibly at the
end points of intervals $I_1, I_2, -I_2$. Fix such $z$ which are not
the end points of intervals $I_1, I_2, -I_2$.

Note that $u^2= m_1(z^2), $  we can choose $u=\sqrt{m_1(z^2)}$ (cf.
equation (\ref{12f}) for an explicit formula of $m_1(z)$) that is
smooth in a neighbourhood of $z$. $g$ is independent of such choice.
It is also clear that $g''$ exists and is bounded on such points,
since $m_1(z)$ is a smooth rational function  on the unit circle
and $|m_1(z)|=1$. Similarly $h$ verify Lemma \ref{lemmaprop2}.
\end{proof}

\subsection{Hankel operator}
Let $K_1 = L^2(I_1) $, $K_2 = L^2(I_1\cup I_2\cup I_{-2}) $, $F  =
L^2(I_1) + jL^2(I_1)$ as in the beginning of section \ref{general}.
We shall prove that $(1 - P)P_F P \in \L^q$ for $\frac23 < q < 1$,
where $\L^q$ denotes the class of von Neumann-Schatten operators.
Recall that $\L^q$ is an ideal which consists of bounded operators
$T$ with
\[
||T||_q:= \big(\sum_n |\lambda_n|^q\big)^{\frac{1}{q}}< \infty \ ,
\]
where $\lambda_n$ are the singular values of $T$ (cf. Sect. \ref{Angle}).
\begin{lemma}\label{lemma1}
If $(1 - P)P_F P \in \L^q$, $0< q < 1$, then the von Neumann entropy
$S(\s) < \infty$, where $\s = P_F P_{F'} P_F$ .
\end{lemma}
\begin{proof}
We note that
\begin{align*}
P_F P_{F'}P_F &= P_F((2P -1)(1 - P_F)(2P- 1)P_F\\
&= 4[P_F , P](1 - P_F) [P, P_F]
\end{align*}
and
\[
[P_F , P] = (1 -P) P_F P - P P_F(1-P) = (1 -P) P_F P -\big((1 -P) P_F P\big)^*.
\]
Hence, if $(1 -P) P_F P\in \L^q$, then $ P_F P_{F'}P_F\in \L^q$.

Let $\{\l_n\}$ be the singular values of $P_F P_{F'}P_F$; then
$S(P_F P_{F'}P_F) = -\sum_n\l_n\log\l_n$. Since $\sum_n\l_n^q
<\infty$, and $-x\log x < x^q$ when $x$ is close to zero, the lemma
is proved.
\end{proof}
By Lemma \ref{lemma1}, we need to look at
\[
(1 -P) P_F P = (1 -P) P_{12} P = (1-P)M_g P + (1-P)M_h RP\, .
\]
Note that $RP = PR = (1-P)M_g P + (1-P)M_h RP$.

In terms of the basis ${z^n}$ of $L^2(S^1)$, we have
\[
P M_g (1-P)(z^n) = \sum_{k\geq 0} g_{k-n}z^k \ ,
\]
these are Hankel operators (cf. \cite{Peller}).
\begin{theorem}\label{th2}
Suppose $f =\sum_n f_n z^n$, $f_n = O(n^{-\a})$ with $\a > \frac32$.
Then $P M_f (1-P)\in \L^q$ with $1 > q > \frac1{\a -\frac12}$.
\end{theorem}
\begin{proof}
(First proof.) As in \cite{How}, let $\xi_n = \sum_{k \geq 0}
f_{k+n}z^k$, $ n>0$,
\[
P M_f (1-P)(z^n) = \sum_{k \geq 0} f_{k-n}z^k = \xi_{-n},\quad n<0\ .
\]
It follows that
\[
P M_f (1-P) = \sum_{n< 0} (\cdot\, ,z^n )\xi_{-n} =
\sum_{n> 0} (\cdot\, ,z^{-n}) \xi_{n}\ .
\]
By \cite[Appendix 1]{Peller}, we have $||T + R||^q_q \leq ||T ||^q_q + ||R||^q_q$ for $0 < q < 1$.
It follows that
\[
||P M_f (1-P)||^q_q \leq \sum_{n\geq 0} ||\xi_n||^q\ .
\]
Note that $||\xi_n|| = \big(\sum_{k\geq 0}|f_{k+n}|^2\big)^\frac12 =
O(n^{-\a + \frac12})$, so the theorem follows.
\smallskip

\indent (Second proof.) According to Page 243 of \cite{Peller} $P
M_f (1-P)\in \L^q$ if the function $f^1(z) = \sum_{n\geq 0}f_{n+1}
z^n$ is in the Besoz space $B_q$, i.e.
\[
\big( 1 - |z|^2\big)^k {f^1}^{(k)} (z) \in L^p(\mathbb D, dA/(1 - |z|^2 ) \ ,
\]
for any integer $pk>1$. Here $dA$ is the Lebesgue measure on the disk $\mathbb D$.

Let us show that if $f_n = O(n^{-\a})$, then $f^1\in B_q$.

Let $k = k_1 + k_2$ with $qk_1 >1$, $0<q<1$. Then, by the Cauchy-Schwartz inequality, we have
\begin{multline*}
\iint_{\mathbb D}(1 -|z|^2)^{k_1 q}\big((1 -|z|^2)^{k_2}| {f^1}^{(k)}(z)\big)^q   \frac1{(1-|z|^2)^2} dA
\\ \leq
\left(\iint_{\mathbb D}\frac{(1 -|z|^2)^{k_1 q}}{(1 - |z|^2)^2}
\big( (1 -|z|^2)^{k_2}| {f^1}^{(k)}(z)|\big)^2dA\right)^\frac12
\left(\iint_{\mathbb D}(1 -|z|^2)^{k_1 q - 2}dA\right)^\frac12 \ .
\end{multline*}
It is sufficient to check that
\[
\iint_{\mathbb D}(1 -|z|^2)^{k_1 q - 2}\big((1 -|z|^2)^{2k_2}| {f^1}^{(k)}(z)|\big)^2 dA< \infty\ .
\]
We have
\begin{align*}
\iint_{\mathbb D}(1 -|z|^2)^{k_1 q - 2} & \big((1 -|z|^2)^{2k_2}| {f^1}^{(k)}(z)|\big)^2 dA\\
&\leq c_1 \iint_{\mathbb D}(1 -|z|^2)^{k_1 q - 2 + 2 k_2}\sum_{n\geq 1} n^{2k - 2\a}|z|^{2n}dA\\
&\leq c_2 \sum_{n\geq 1} n^{2k - 2\a} \cdot  \iint_{\mathbb D}(1 -|z|^2)^{k_1 q - 2 + 2 k_2} dA
\iint_{\mathbb D}(1 -|z|^2)^{k_1 q - 2 + 2 k_2}dA \ .
\end{align*}
Using polar coordinates, we have
\begin{align*}
\iint_{\mathbb D}(1 -|z|^2)^{k_1 q - 2 + 2 k_2}|z|^{2n}dA
&= 2\pi \int_0^1 (1 -r^2)^{k_1 q - 2 + 2 k_2} r^{2n} rdr \\
&= \frac{2\pi}{2} \int_0^1 (1 -r^2)^{k_1 q - 2 + 2 k_2} r^{n} dr  \\
&= \pi\,  \frac{\Ga( {k_1 q + 2 k_2 -1)}\Ga(n+1)   }{  \Ga (k_1 q + 2 k_2 -1 + n -1) }\ .
\end{align*}
By using Sterling's formula $\Ga(n+1) \sim \sqrt{2\pi n}
\left(\frac{n}{e}\right)^n$, we have
\[
\frac{\Ga( {k_1 q + 2 k_2 -1)}\Ga(n+1)   }{  \Ga (k_1 q + 2 k_2 -1 + n -1) }\leq c_3 \,n^{1 - k_1 q - 2k_2}
\]
when $n\to\infty$.

It follows that
\[
\iint_{\mathbb D}(1 -|z|^2)^{k_1 q - 2}  \big((1 -|z|^2)^{2k_2}| {f^1}^{(k)}(z)|\big)^2 dA
\leq c_4  \sum_{nn\geq 1} n^{2k - 2\a +1  - k_1 q - 2k_2} < \infty
\]
if $-k_1 q + 2k_2 -1  + 2 \a - 2k >1$. Note that $k_1q > 1$.

These two inequalities imply that
\[
\frac1q < k_1 < \frac{2\a -2}{2-q} \ ,
\]
which is possible iff $  \frac{2\a -2}{2-q} - \frac1q > 0$ iff $q > \frac1{\a - \frac12}$\, .
\end{proof}

\subsection{Finiteness of von Neumann entropy for free Fermi nets}
Let $K_1 = L^2(I_1)$, $K_2 = L^2(I_1\cup I_2 \cup -I_2)$, $F =
L^2(I_1) + j L^2(I_1)$  as in the beginning of section
\ref{general}. Note that we have the inclusion
\[
\M(K_1) \subset \M(F) \subset \M(K_2) \ .
\]
\begin{theorem}\label{th3}
\begin{itemize}
\item[$(1)$]  $(1 - P)P_F P \in \L^q$ for $\frac23 < q < 1$ \ .
\item[$(2)$]
The von Neumann entropy associated with $\M(F), \Omega$, denote by
$S(F)$, is finite.
\end{itemize}
\end{theorem}
\begin{proof}
$(1)$ follows by Theorem \ref{th11}, Proposition \ref{lp2} and
Theorem \ref{th2}. $(2)$ follows from $(1)$, Lemma \ref{lemma1} and Corollary \ref{S(a)}.
\end{proof}
\begin{remark}
(1) It is an interesting question to calculate the entropy in Th.
\ref{th3}. Note that by Th. 3.18  of \cite{LXu} that the  entropy in
Th. \ref{th3} is bounded from below by $\frac{1}{6}\ln \frac{1}{\cos
(\phi/2)}$. We expect that the theory of Hankel operators as in
\cite{Peller} will be useful in solving this question. For heuristic
computations using replica methods in some holographic models, see
\cite{Fal}.

\par

(2)  Since any open intervals $I,  \tilde I$ with $\bar{I}\subset \tilde I$ can be
mapped to the symmetric intervals  in the above theorem by elements
in $\mob$, it follows that the above theorem is also true for $K_1
=L^2(I)$, $K_2 = L^2(\tilde I).$ \end{remark}

\subsubsection{Real fermion case}
Define conjugate linear operator $Q= M_{z^{-1}} C$ on $\H$ where $C$
is complex  conjugation and  $M_{z^{-1}}$ is multiplication by
$z^{-1}$. We have $(Qf)(z) = z^{-1} \overline{f(z)}.$ On basis $z^n$
of $\H$ we have \ben\label {10g} Q(z^n)= z^{-n-1} \, . \een

\begin{lemma}\label{lemr}
(1) $Q^2=I$ and $Q$ commutes with $U_m$ as defined in equation
(\ref{10f}) and flip operator  $F_1$ as defined after  equation
(\ref{10f});
\par (2) $QPQ= 1-P, Q i(2P-I) Q =  i(2P-I)$;
\par (3) $QM_fQ = M_{\overline{f}}$ and $Q$ commutes with $P_{12}$ as in
Th. \ref{th3}.

\end{lemma}
\begin{proof}
(1) follows from definitions. For (2), it is enough to check that
 $QPQ= 1-P$ since $Q$ is conjugate linear. One checks easily that
  $QPQ= 1-P$ on $z^n$ from equation (\ref{10g}). To prove (3), note
  that $\overline{g} =g, \overline{h}= -h$ where $g,h$  are as in
  Th. \ref{th3}. (3) follows from $CM_g C= M_g, \ \ CM_h C= -M_h, \
  CR=RC, \
  \
  RM_{z^{-1}}R = -M_{z^{-1}}$.
\end{proof}
Denote by $Q_{\pm 1}:= \frac{1}{2}(I\pm Q)$ the projections from
$\H$ to the eigenspaces $\H_{\pm 1}$ of $Q$ with eigenvalues $\pm
1$. Note that $\H =\H_{+1} \bigoplus \H_{-1}$ while by slightly
abuse of notations in this section only we use $\bigoplus$ to mean
orthogonal with respect to the real part of inner product on $\H$.
Note that $M_{-i}\H_{+1} =\H_{-1}$ since $Q$ is conjugate linear.

For an interval $I$ we let $L^2(I)_{\pm 1}= L^2(I)\cap \H_{\pm 1}.$
The real free fermion net associates to interval $I$ the von neumann
algebra $\M( L^2(I)_{+ 1})$ generated by $a(f)+ a(f)^*, \forall f\in
L^2(I)_{+ 1}.$ We note that $\M( L^2(I))$  is the graded tensor
product of $\M( L^2(I)_{+ 1})$  and $\M( L^2(I)_{- 1})$.

Let $K_1 = L^2(I_1)$, $K_2 = L^2(I_1\cup I_2 \cup -I_2)$, $F =
L^2(I_1) + j L^2(I_1)$  as in the beginning of section
\ref{general}. Note by Lemma \ref{lemr} $K_1, K_2, F$ are
$Q$-invariant. Let $(K_1)_{\pm 1}:= K_1\cap \H_{\pm 1}, (K_2)_{\pm
1}:= K_2\cap \H_{\pm 1}, F_{\pm 1}:= F\cap \H_{\pm 1}$. Then we have
$\M((K_1)_{\pm 1})\subset \M(F_{\pm 1})\subset \M((K_2)_{\pm 1}),$
where $ \M(F_{\pm 1})$ is the canonical type $I$ factor.  Note that
we have $S(F)= S(F_{+1}) + S(F_{-1})$. Since $\M(F_{- 1}) = \Lambda
(-i) \M(F_{- 1}) \Lambda (i)$ and the vacuum is preserved by
$\Lambda(i)$, it follows that $S(F_{+1}) = S(F_{-1})= \frac{1}{2}
S(F)$. We have therefore proved the following:

\begin{theorem}\label{th3r}
We have $\M((K_1)_{\pm 1})\subset \M(F_{\pm 1})\subset \M((K_2)_{\pm
1}), $ and
 $S(F_{+1}) = S(F_{-1})=
\frac{1}{2} S(F) < \infty $ where $S(F)$ is the  von Neumann entropy
in Th. \ref{th3}.
\end{theorem}

\subsection{Passing to free bosons}
\label{bos}

Note first that
\ben\label{Fcan}
\M(F) = \M\big(L^2(I_1)\big) \vee kJ_{\M(K_1)^c\cap \M(K_2)} \M\big(L^2(I_1)\big) J_{\M(K_1)^c\cap
\M(K_2)} k^{-1} \ ,
\een
where $\M(K_1)^c$ denotes the graded commutant of $\M(K_1)$, and $k$
is the Klein transformation.

Let $g$ be an automorphism of the net $\M(L^2(I))$ implemented
by a unitary $U(g)$ on $\La(\H)$ which commutes with grading
operator.  This means  $U(g)\Omega = \Omega$ and $U(g) \M(L^2(I))
U(g)^* = \M(L^2(I))$ and $U(g)$ commutes with $\Gamma=k^2$ with $k$
the Klein transformation.  It is clear that such adjoint action of
$U(g)$ preserves $\M(F)$ as a set.

Example of such automorphisms are $\Ga = k^2$ with $k$ the Klein
transformation and more generally  the automorphisms implemented by
the multiplication by $\l$, $|\l| = 1$, on $L^2(I)$. This is the
action of $U(1)$ on free fermion net. In fact one can also check
from Th. \ref{th11} that the automorphisms implemented by this
$U(1)$  action preserves $\M(F)$ as a set.

In the notation of \cite{X1}, the free fermion net is written as
$\A_{U(1)_1}$. The free bosonic net with central charge $c =1$ is
the $U(1)$  the fixed point net of $\A_{U(1)_1}$ under the action of
$U(1)$ as described above (cf. Page 186 of \cite{X1}). Denote by $G
:= U(1)$. Then $U(g) \M_F U(g)^* = \M(F)$, $g\in G$.

We have
\[
\M(K_1)^G \subset \M(F)^G \subset \M(K_2)^G \ ,
\]
where $\M(F)^G$ is the fixed point of $\M(F)$ under the action of $G$. As $\M(F)$ is a type $I$ factor and $G$ is compact, $\M(F)^G$ is a type $I$ von Neumann algebra, but $\M(F)^G$ is not a factor, rather $\M(F)^G$ is the direct sum of countably many type $I$ factors.

Since the vacuum state $\omega = (\Omega , \cdot\Omega)$ is
preserved by the action of $G$, by Proposition 6.7 of \cite{OP} we
have
\[
S(\omega|_{\M(F)^G}) \leq S(\omega|_{\M(F)}) < \infty \ ,
\]
where last inequality follows by Theorem \ref{th3}. By the remark
(2)  after  Theorem \ref{th3} we have therefore proved the
following:
\begin{theorem}\label{th4}
Let $I, \tilde I$ be open intervals with $\bar{I}\subset \tilde I$.  For the free
bosonic net $\A^G_{U(1)_1}$ we have that
\[
\A^G_{U(1)_1}(I) \subset \M(F)^G\subset \A^G_{U(1)_1}(\tilde I) \ ,
\]
where $\M(F)^G$ is a type $I$ discrete, von Neumann algebra and the vacuum von Neumann entropy
\[
S\big(\omega |_{\M(F)^G}\big) < \infty \ .
\]
\end{theorem}
\begin{remark}
By taking the graded tensor product of $r$-copies of the net
$\A^G_{U(1)_1}$, Theorems \ref{th3} and \ref{th4} immediately
generalise to the case of the conformal net of $r$ free fermions.
\end{remark}

\section{Conclusion} Let $O_R$ denote the double cone in the Minkowski spacetime that is the causal envelop of a time zero ball of radius $R>0$ centered at the origin. We have shown a natural way to define rigorously the entanglement entropy $S_\A(O_{R'}, O_{R})$ of a QFT net of von Neumann algebras $\A$ associated with an inclusion $O_{R'} \subset O_{R}$, $R' < R$.  We have shown that $S_\A(O_{R'}, O_{R})$ is finite, in particular in low dimensional free models.

It will be natural to establish that this entropy is finite, also for higher dimensional models, in a model independent ground. We expect that $S_\A(O_{R'}, O_{R})$ is finite if the modular nuclearity condition holds, see \cite{Haag}.

As $R'\to R$, the entropy $S_\A(O_{R'}, O_{R})$ should diverge. However, the divergence leading term should grow proportionally to the area of the sphere of radius $R$, according to a seemingly universal feature (area theorems).

Results in this direction would put several fundamental issues in theoretical physics on a rigorous mathematical ground.

\bigskip

\noindent
{\bf Acknowledgements.}
Our collaboration took place, in particular, during the CERN workshop on ``Advances in Quantum Field Theory" in March-April 2019 and the program ``Operator Algebras and Quantum Physics'' at the Simons Center for Geometry and Physics at Stony Brook in June 2019. We are grateful to both institutions for the invitations.  

\medskip
\noindent
R.L. acknowledges the MIUR Excellence Department Project awarded to the Department of Mathematics, University of Rome Tor Vergata, CUP E83C18000100006.

\end{document}